\documentclass[a4paper,unpublished]{quantumarticle}
\pdfoutput=1
\usepackage[utf8]{inputenc}
\usepackage[english]{babel}
\usepackage[T1]{fontenc}
\usepackage{amsmath,amssymb,amsthm}
\usepackage{mathtools}
\usepackage{graphicx}
\usepackage{booktabs}
\usepackage{hyperref}
\usepackage{xcolor}
\newcommand{\newtext}[1]{#1}

\newtheorem{theorem}{Theorem}
\newtheorem{lemma}[theorem]{Lemma}
\newtheorem{proposition}[theorem]{Proposition}
\newtheorem{corollary}[theorem]{Corollary}
\newtheorem{conjecture}[theorem]{Conjecture}
\theoremstyle{remark}

\newcommand{\THat}{\mathcal{T}_{\mathrm{Hat}}}

\newcommand{\OHat}{\Omega_{\mathrm{Hat}}}
\newcommand{\CHat}{\mathcal{C}_{\mathrm{Hat}}}
\newcommand{\CSpec}{\mathcal{C}_{\mathrm{Spec}}}
\newcommand{\CC}{\mathbb{C}}
\newcommand{\ZZ}{\mathbb{Z}}
\newcommand{\RR}{\mathbb{R}}
\newcommand{\Etwo}{\mathrm{E}(2)}
\newcommand{\SEtwo}{\mathrm{SE}(2)}
\newcommand{\freq}{\operatorname{freq}}
\newcommand{\rhoK}{\rho_{K}}
\newcommand{\gphi}{\varphi}
\newcommand{\LI}{\mathrm{LI}}

\begin{document}

\title{Quantum error-correcting codes from aperiodic monotiles:
the Hat and the Spectre}

\author{Josep Batle}
\email{batlequantum@gmail.com}
\affiliation{CRISP -- Centre de Recerca Independent de sa Pobla,
07420 sa Pobla, Balearic Islands, Spain}
\author{Adam Bednorz}
\affiliation{Faculty of Physics, University of Warsaw,
ul.\ Pasteura 5, 02-093 Warsaw, Poland}

\maketitle

\begin{abstract}
Li and Boyle showed that the Penrose tiling defines a quantum
error-correcting code: superpositions of tilings over isometry orbits
protect quantum information against erasure of any bounded region. We
extend the construction to the aperiodic monotiles discovered by
Smith, Myers, Kaplan and Goodman-Strauss. For the Hat, we prove
that every Hat tiling belongs to the substitution hull, we prove
strong local indistinguishability for \emph{all} Hat tilings, and we
prove local recoverability unconditionally for all \emph{nonsingular}
Hat tilings via the torus parametrization of the underlying
cut-and-project scheme. The remaining singular case reduces to one
sharply posed geometric question --- can a region that is a union of
hats be retiled a second way? --- which we verify computationally has
no counterexample up to a substantial scale: a certified $2490$-tile
patch admits precisely one tiling by hats, so all $2^{2490}$ of its
tile-subregions retile uniquely. Unlike the Penrose, Ammann--Beenker
and Fibonacci tilings, both monotiles form \emph{two}
local-indistinguishability classes, so their code spaces split into
two erasure-correcting sectors carrying a superselected classical
label. Whether the label survives depends on which isometries are
gauged: the Spectre's classes are exchanged by a $30^{\circ}$
rotation and merge once all proper isometries are gauged, whereas the
Hat's are exchanged only by reflections. Under the physically natural
gauge group $\SEtwo$, the Hat code therefore stores one robust
classical bit --- the handedness of its long-range order, read from
a window of a few tiles, with the reflected-tile fraction separated by
$\Delta = |K|\sqrt{5}/3$ between the sectors --- alongside its
protected quantum sectors. It is the reflexible
monotile, not the chiral one, that carries the chirality bit. We give
the exact Perron--Frobenius data of both codes, including the
per-class reflected-Hat frequencies $(3\mp\sqrt{5})/6$ and the
Spectre orientation-class frequencies $(5\pm\sqrt{15})/10$.
\end{abstract}

\section{Introduction}
\label{sec:intro}

Quantum error correction rests on storing logical information in
degrees of freedom that no local observer can read. The canonical
constructions --- stabilizer codes, topological codes --- engineer
this locally hidden information by hand. Li and Boyle
\cite{LiBoyle2024} observed that nature supplies examples for free:
the Penrose tiling, viewed as a quantum state, already has the two
properties that erasure correction requires. First, any two Penrose
tilings are \emph{locally indistinguishable}: every finite pattern
occurs in each of them with identical frequency, so no bounded
measurement can tell them apart. Second, they are \emph{locally
recoverable}: the tiles deleted from any bounded region are uniquely
determined by the tiles outside it. Superposing a tiling over its
orbit under the isometry group then produces codewords whose reduced
density matrices on any bounded region coincide, and whose mutual
distinguishability by local operators vanishes --- the Knill--Laflamme
conditions for erasure correction \cite{LiBoyle2024}. The same
construction applies to the Ammann--Beenker tiling and, in one
dimension, to the Fibonacci quasilattice.

The discovery in 2023 of aperiodic \emph{monotiles} transformed the
landscape of aperiodic order. Smith, Myers, Kaplan and Goodman-Strauss
first exhibited the Hat, a single polykite that tiles the plane only
aperiodically when copies of both handednesses are allowed
\cite{Smith2024hat}, and shortly afterwards the Spectre, a strictly
chiral monotile whose tilings use one handedness only
\cite{Smith2024spectre}; both shapes are drawn in
Figure~\ref{fig:anatomy}. The associated tiling spaces were rapidly
placed on the same footing as the classical quasicrystals: the Hat
family is topologically conjugate to a self-similar model set (the
CAP tiling) with pure-point spectrum \cite{BGS2025}, and the Spectre
hull was shown to have pure-point dynamical spectrum arising from a
cut-and-project scheme with Rauzy-fractal-type windows
\cite{BGMS2025}.

This paper carries the Li--Boyle construction to both monotiles. For
the Hat the outcome is structural consolidation: we show that both
Li--Boyle hypotheses can be grounded in the discovery paper's own
central theorems --- that \emph{every} Hat tiling carries a unique
infinite hierarchy of supertiles \cite[Thms.~4.1 and
5.1]{Smith2024hat} --- so that the code is defined on the full family
of Hat tilings, not on a substitution-generated subfamily\newtext{;
the two families are in fact the same (Theorem~\ref{thm:hull}), which
settles a question left open in \cite{Smith2024hat}}. Local
indistinguishability follows unconditionally, in its correct per-class
form (Theorem~\ref{thm:LI}). Local recoverability we
prove unconditionally for all \emph{nonsingular} Hat tilings
(Theorem~\ref{thm:recover}) --- the same scope as Li and Boyle's own
Fibonacci code --- using the torus parametrization of the
cut-and-project scheme established in \cite{BGS2025}. The singular case reduces to a single sharply posed geometric
question (Conjecture~\ref{lem:defect}): can a region that is a union of
hats be tiled by hats a second way? We show this question has no counterexample up to a substantial
scale: a certified $2490$-tile patch admits precisely one tiling by
hats, so every one of its $2^{2490}$ tile-subregions retiles uniquely
(Proposition~\ref{prop:computational})\newtext{, and that every
region of radius $18$ kisrhombille units occurring in a level-7
substitution patch of $372\,100$ hats retiles uniquely
(Section~\ref{sec:hat-singular})}.

The second outcome is what we regard as the paper's central
conceptual point: monotile hulls split into local-indistinguishability
sectors that behave as superselection sectors, producing a hybrid
quantum--classical memory --- a structure absent from every
previously studied example. Penrose, Ammann--Beenker and Fibonacci tilings each
form a single local-indistinguishability (LI) class; both monotiles
form \emph{two}. For the Hat these are the chirality classes ---
hat-majority versus anti-hat-majority tilings, exchanged by
reflections, since a mirror image inverts the $\gphi^{4}{:}1$
handedness ratio. For the Spectre they are the two rotational
subclasses of Baake, G\"ahler, Maz\'a\v{c} and Sadun, exchanged by
a $30^{\circ}$ rotation \cite{BGMS2025}. Each class separately supports the Li--Boyle construction, while a
coarse local observable distinguishes the classes: the code space
splits into two erasure-correcting sectors carrying a superselected
classical label. Whether the label survives depends on which
isometries are gauged,
and here the two monotiles part ways. The Spectre's classes are
exchanged by a \emph{proper} isometry, so gauging $\SEtwo$ ---
position and orientation unobservable, parity not gauged --- merges
its sectors; the Hat's classes are exchanged only by reflections, so
its label survives. Under the physically natural gauge group, an
infinite Hat tiling superselects one robust classical bit: the global
handedness of its long-range order, \newtext{read from a window of a
few tiles through the handedness of a touching pair, and statistically
through the reflected-tile fraction, whose two values are separated
by} $\Delta = |K|\sqrt{5}/3$. Counterintuitively, it is the reflexible
monotile, not the chiral one, that carries the chirality bit: strict
chirality removes the mirrored species from Spectre tilings entirely,
leaving parity nothing to protect.

The paper proceeds in four movements. Section~\ref{sec:framework}
recalls the Li--Boyle framework and fixes the vocabulary of both
fields. Section~\ref{sec:hat} develops the Hat: we first establish
the structural inputs the framework needs, then prove recoverability
in its correct scope and isolate the single question that remains
open, and finally construct the code and explain why the Hat
unexpectedly stores an additional classical bit.
Section~\ref{sec:spectre} runs the same programme for the Spectre
and shows why its analogous bit is gauged away by a mere rotation.
Section~\ref{sec:discussion} discusses the code-dimension question
and open problems. Two reading paths are possible. A reader from
quantum information needs, from the tiling side, only the facts of
\S\ref{sec:hat-pf} and Figure~\ref{fig:anatomy}, and may take the
cited structure theorems on faith at a first pass; a reader from
aperiodic order needs, from the quantum side, only
Eq.~\eqref{eq:KL} and the two hypotheses of
Section~\ref{sec:framework}. Table~\ref{tab:glossary} fixes the
vocabulary of both fields, and
Figures~\ref{fig:anatomy}--\ref{fig:spectre} carry the logical
skeleton of the argument.

\section{The Li--Boyle framework}
\label{sec:framework}

Let $\mathcal{T}$ be a family of tilings of $\RR^{2}$, closed under a
group $G$ of isometries ($G = \Etwo$ or $\SEtwo$ below). To each
tiling $T$ associate a state $|T\rangle$ in an ambient Hilbert space
in which distinct tilings are orthogonal, and define the
group-averaged states
\begin{equation}
|\Psi_{[T]}\rangle \;=\; \int_{G} dg\, |gT\rangle ,
\label{eq:superposition}
\end{equation}
labelled by the orbits $[T] \in \mathcal{T}/G$. Written out,
Eq.~\eqref{eq:superposition} is an equal-weight superposition of the
\emph{same} infinite tiling in every position and orientation the
plane allows,
\begin{equation}
|\Psi_{[T]}\rangle \;\propto\;
\bigl|\, \mathcal{T} \,\bigr\rangle
\;+\; \bigl|\, \mathcal{T}_{+a} \,\bigr\rangle
\;+\; \bigl|\, \rotatebox[origin=c]{22}{$\mathcal{T}$} \,\bigr\rangle
\;+\; \bigl|\, {\rotatebox[origin=c]{-38}{$\mathcal{T}$}}_{+a'}
\,\bigr\rangle \;+\; \cdots ,
\label{eq:ketsum}
\end{equation}
where the tilt and the subscript of the glyph indicate the rotation
and translation applied to the pattern. No term is preferred, because
no external frame exists to prefer it: $|\Psi_{[T]}\rangle$ is the
state of the tiling with completely delocalized position and
orientation --- the planar analogue of a zero-momentum,
zero-angular-momentum state. All physical information resides in the
internal relations of the pattern; none resides in where or how it
sits. The candidate code is
$\mathcal{C} = \overline{\operatorname{span}}
\{ |\Psi_{[T]}\rangle \}$. Erasure of a bounded region
$K \subset \RR^{2}$ is correctable if and only if the Knill--Laflamme
conditions hold: for every operator $O_{K}$ supported in $K$,
\begin{equation}
\langle \Psi_{[T_{1}]} | O_{K} | \Psi_{[T_{2}]} \rangle
\;=\; c(O_{K})\, \delta_{[T_{1}],[T_{2}]} .
\label{eq:KL}
\end{equation}
Li and Boyle \cite{LiBoyle2024} isolate two hypotheses on
$\mathcal{T}$ that together imply Eq.~\eqref{eq:KL}:

\emph{(H1) Strong local indistinguishability.} Every finite patch $P$
has a well-defined frequency $\freq(P)$, identical in every
$T \in \mathcal{T}$. This makes the diagonal value in
Eq.~\eqref{eq:KL} independent of $[T]$: the reduced density matrix
$\rhoK = \operatorname{Tr}_{K^{c}} |\Psi_{[T]}\rangle\langle
\Psi_{[T]}|$ depends only on the patch statistics visible in $K$,
which are the same for all representatives. The group average is what
turns statistics into states: for the bare tiling state $|T\rangle$,
$\rhoK$ would display the literal patch at $K$, which differs from
tiling to tiling; for the averaged $|\Psi_{[T]}\rangle$ it is the
frequency-weighted mixture over all patches, and (H1) makes that
mixture independent of the representative.

\emph{(H2) Local recoverability.} For every bounded $K$, the
restriction map $T \mapsto T_{K^{c}}$ is injective on $\mathcal{T}$:
the deleted patch is uniquely determined by its complement, so a
recovery operation can regenerate the lost region from the intact
remainder. This
forces the off-diagonal terms to vanish: if $[T_{1}] \neq [T_{2}]$,
then $g T_{1}$ and $h T_{2}$ differ outside $K$ for all $g, h \in G$,
so $O_{K}$ cannot connect them.

For the Penrose tiling, (H1) follows from unique ergodicity of the
hull and (H2) from the uniqueness of Robinson composition
\cite{LiBoyle2024}. Sections~\ref{sec:hat} and \ref{sec:spectre}
establish the analogues for the Hat and the Spectre.

One further degree of freedom in the construction will matter below:
the choice of the group $G$. Any group of isometries under which
$\mathcal{T}$ is closed can be used, and the choice is physical ---
it declares which reference frames are deemed unobservable. Averaging
over $\SEtwo$ gauges position and orientation but not parity;
averaging over $\Etwo$ gauges parity as well. Sections
\ref{sec:hat-code} and \ref{sec:spec-qubit} show that the sector
structure of the monotile codes depends on this choice in an
essential way.

Table~\ref{tab:glossary} fixes the vocabulary of both fields as used
in this paper.

\begin{table*}[t]
\centering
\small
\renewcommand{\arraystretch}{1.18}
\begin{tabular}{@{}lp{0.74\textwidth}@{}}
\toprule
\multicolumn{2}{@{}l}{\emph{Aperiodic order}}\\
\midrule
patch & a finite configuration of tiles occurring in a tiling. \\
patch frequency & number of copies of a patch per unit area, in the
infinite-volume limit. \\
LI class & maximal family of tilings sharing all patch frequencies;
its members cannot be distinguished by any bounded window. \\
hull $\Omega$ & closure of the translation orbit of a tiling in the
local topology; here each chirality (resp.\ orientation) class is a
hull. \\
inflation (substitution) & replacing each tile by a fixed
configuration of smaller tiles and rescaling; the linear inflation
factor is the rescaling ($\gphi^{2}$ for the Hat). \\
deflation (composition) & the inverse: grouping tiles into
supertiles; unique for both monotiles (unique hierarchy). \\
cut-and-project & construction of the tiling as a slice through a
higher-dimensional lattice; the slice position is the tiling's phase
$\gamma$. \\
window & the acceptance domain in the internal space of the
cut-and-project scheme. \\
singular tiling & one whose phase meets the window boundary; the
finitely many tilings sharing that phase form its fiber. \\
\midrule
\multicolumn{2}{@{}l}{\emph{Quantum information}}\\
\midrule
code space & subspace whose states carry the protected
information. \\
erasure of $K$ & loss of all degrees of freedom in the region $K$,
with the location of $K$ known. \\
Knill--Laflamme condition & Eq.~\eqref{eq:KL}; necessary and
sufficient for correctability of erasure of $K$. \\
sector / superselection & block decomposition such that no local
operator connects blocks; relative phases between blocks are
unobservable and the block label behaves classically. \\
hybrid memory & a direct sum of quantum codes indexed by a decodable
classical label. \\
\bottomrule
\end{tabular}
\caption{Working glossary. Each entry states what the term means in
this paper, not its most general usage.}
\label{tab:glossary}
\end{table*}

\section{The Hat code}
\label{sec:hat}

\subsection{Structural preliminaries}
\label{sec:hat-pf}

The hat is an 8-kite polykite on the kisrhombille lattice
\cite{Smith2024hat}. Let $\THat$ denote the set of all tilings of the
plane by isometric copies of the hat and its mirror image (the
anti-hat); every Hat tiling necessarily contains both handednesses.
The relevant symmetry group is the full Euclidean group $\Etwo$,
which acts on $\THat$.

\newtext{The structural input from the discovery paper comes in two
theorems, both proved by a combinatorial and computer-assisted
analysis. Theorem~4.1 of \cite{Smith2024hat} states that every Hat
tiling divides into the four clusters $H,T,P,F$ of
\cite[Fig.~4.1]{Smith2024hat}, satisfying the matching rules given
there, with all clusters unreflected or all reflected; the tiling by
metatiles has the same symmetries as the tiling by hats, and
conversely every tiling by the labelled metatiles induces a Hat
tiling. Theorem~5.1 of \cite{Smith2024hat} states that in any tiling
by the four metatiles obeying the matching rules, the metatiles group
(after bisecting the $P$ and $F$ tiles) into larger supertiles that
are combinatorially equivalent to the metatiles, so that the grouping
can be repeated $n$ times for every $n$. Together they give, for every
$T\in\THat$ and every $n\ge1$, a unique decomposition of $T$ into
level-$n$ supertiles; we refer to this as the \emph{unique-hierarchy
theorem}. It is the tiling-theoretic analogue, established
specifically for the Hat, of Solomyak's general unique-composition
theorem for aperiodic self-similar tilings
\cite{Solomyak1998,Solomyak1997}, with two differences that are
easily overlooked and are used below. First, the supertiles are
combinatorially, not geometrically, equivalent to the metatiles: with
the exception of $T$, no supertile is similar to its metatile, and the
supertile outlines change from level to level
\cite[\S2 and Fig.~2.6]{Smith2024hat}; the self-similar member of the
Hat family is the CAP tiling of \cite{BGS2025}, not the Hat itself.
Second, the composition of hats into metatiles is determined by the
first corona of a hat except for the distinction between the $P$ and
$F$ clusters, which requires one further neighbour
\cite[\S4.1]{Smith2024hat}; it is a rule of radius two coronas, not
one.} The linear inflation of the (period-two, chirality bookkeeping)
substitution is $\gphi^{2}$, with $\gphi = (1+\sqrt{5})/2$.

\paragraph{Two chirality classes.}
Fix once and for all which of the two mirror-image polykites is
called \emph{the hat}; its mirror image is the anti-hat. Reflecting
a Hat tiling produces another Hat tiling (the tile set is
reflection-closed), but \emph{not} one with the same patch
statistics: unreflected and reflected hats occur in the ratio
$\gphi^{4} : 1$ (Proposition~\ref{prop:PF} below), and a mirror
image inverts this ratio. Hat tilings therefore split into exactly
two translational local-indistinguishability classes,
\[
\THat \;=\; \THat^{+} \sqcup \THat^{-},
\]
distinguished by which handedness forms the majority
(Theorem~\ref{thm:LI} below makes this precise), with $\THat^{-}$
the set of mirror images of $\THat^{+}$. Every proper isometry
preserves each class, because rotations and translations do not
change the handedness of a tile; reflections exchange the two. The
class label is \emph{frame-independent}: it refers only to tile
chirality, never to orientation. Both properties are decisive in
\S\ref{sec:hat-code}. Figure~\ref{fig:anatomy} shows both monotiles: the hat and the
anti-hat on their underlying lattice, together with the Spectre,
which is treated in Section~\ref{sec:spectre}.

\begin{figure}[t]
\centering
\includegraphics[width=\linewidth]{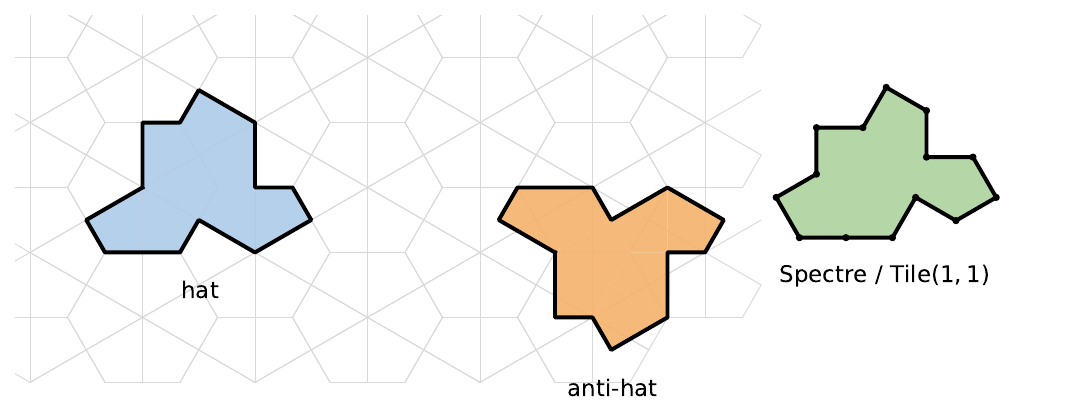}
\caption{The prototiles. The hat (blue) is a union of eight kites of the
kisrhombille lattice (grey); the anti-hat (orange) is its mirror
image. Every Hat tiling uses both, but never in equal proportion: one
handedness always dominates, in the exact ratio $\gphi^{4} : 1$, and
which one dominates is the bit of \S\ref{sec:hat-code}. The Spectre
(green) is the equilateral relative Tile$(1,1)$: the same fourteen
edge directions with both edge lengths set equal --- the two
interleaved direction classes each sum to zero (verified to machine
precision), so the polygon closes for every length pair, which is the
whole Tile$(a,b)$ family in one identity. Its tilings use one
handedness only; in the physical Spectre the straight edges are
replaced by congruent curves so that reflected copies cannot fit
\cite{Smith2024spectre}.}
\label{fig:anatomy}
\end{figure}

For frequency bookkeeping we use a two-type census: a Hat tiling
decomposes into \emph{compounds} (a reflected hat paired with an
unreflected partner) and \emph{single} unreflected hats, with
substitution matrix
\begin{equation}
M \;=\; \begin{pmatrix} 1 & 1 \\ 5 & 6 \end{pmatrix},
\quad
\begin{aligned}
\text{comp.} &\mapsto 1\,\text{comp.} + 5\,\text{sing.},\\
\text{sing.} &\mapsto 1\,\text{comp.} + 6\,\text{sing.}
\end{aligned}
\label{eq:M}
\end{equation}
\newtext{This is the two-cluster substitution of
\cite[Fig.~2.11]{Smith2024hat}: a reflected hat merged with one
neighbour is replaced by one compound and five hats, a single hat by
one compound and six hats.} This census tracks the
reflected/unreflected content; the full four-metatile substitution of
\cite{Smith2024hat} refines it but is not needed for the frequency
statements below. \newtext{The two agree: the four-metatile
substitution has tile frequencies $(5/3-\gphi,\,1/3,\,2\gphi-3,\,
2-\gphi)$ for $(T,H,P,F)$ \cite[\S3]{BGS2025}, the $H$ metatile is the
only one containing a reflected hat, one of its four, and the
metatiles carry $1,4,2,2$ hats, so the reflected fraction is
$(1/3)/(1+\gphi)=1/(3\gphi^{2})=(3-\sqrt5)/6$, as in
Eq.~\eqref{eq:fr} below.}

\begin{proposition}[Substitution dynamics of the census matrix]
\label{prop:PF}
Let $M$ be the matrix of Eq.~\eqref{eq:M}. Then:
\begin{enumerate}
\item $M$ is primitive; indeed every entry of $M$ itself is strictly
positive.
\item The characteristic polynomial is $x^{2} - 7x + 1$, with
eigenvalues
\[
\lambda_{\pm} \;=\; \frac{7 \pm 3\sqrt{5}}{2} \;=\; \gphi^{\pm 4},
\]
so the Perron--Frobenius eigenvalue $\lambda_{+} = \gphi^{4} \approx
6.854$ equals the area inflation, consistent with linear inflation
$\gphi^{2}$.
\item The Perron--Frobenius eigenvector, normalized so that the
single-hat component equals 1, is
\[
(v_{c},\, v_{s}) \;=\; \Bigl( \tfrac{3\sqrt{5}-5}{10},\ 1 \Bigr)
\;\approx\; (0.170820,\ 1).
\]
\item Since each compound contains exactly one reflected hat, the
asymptotic reflected-hat frequency is
\begin{equation}
f_{r} \;=\; \frac{v_{c}}{2 v_{c} + v_{s}}
\;=\; \frac{3 - \sqrt{5}}{6} \;\approx\; 0.127322 ,
\label{eq:fr}
\end{equation}
equivalently $(1-f_{r})/f_{r} = \gphi^{4}$: unreflected and reflected
hats occur in the ratio $\gphi^{4} : 1$.
\item For every strictly positive initial census vector $v$, the
normalized iterates $M^{n} v / \lVert M^{n} v \rVert$ converge to the
Perron--Frobenius direction, and the reflected fraction converges to
$f_{r}$ exponentially with the exact rate
\begin{equation}
\begin{aligned}
\bigl| f_{r}^{(n)} - f_{r} \bigr|
&\;=\; O\!\left( \bigl( \lambda_{-}/\lambda_{+} \bigr)^{n} \right)
\;=\; O\!\left( \gphi^{-8n} \right),\\[2pt]
\gphi^{-8} &\;\approx\; 0.0212862 .
\end{aligned}
\label{eq:rate}
\end{equation}
\end{enumerate}
\end{proposition}

\begin{proof}
(1) is immediate from Eq.~\eqref{eq:M}. (2) and (3) are direct
computations: $\operatorname{tr} M = 7$, $\det M = 1$, and
$M (v_{c}, 1)^{\mathsf T} = \gphi^{4} (v_{c}, 1)^{\mathsf T}$ is
checked by $v_{c} + 1 = \gphi^{4} v_{c}$ together with
$5 v_{c} + 6 = \gphi^{4}$. (4) counts one reflected hat per compound
among $2 v_{c} + v_{s}$ hats per census unit; the ratio identity
follows from $(v_{c} + v_{s})/v_{c} = (7 + 3\sqrt{5})/2$. (5) is the
quantitative Perron--Frobenius theorem for a diagonalizable positive
matrix: the error contracts by the eigenvalue ratio
$\lambda_{-}/\lambda_{+} = \gphi^{-8}$ per substitution step.
\end{proof}

Figure~\ref{fig:pf} shows the numerical verification: convergence of
$f_{r}^{(n)}$ to $(3-\sqrt{5})/6$ from three unrelated seeds, with the
measured per-step error ratio matching $\gphi^{-8} = 0.0212862\ldots$
to eight digits by generation 5.

\begin{figure}[t]
\centering
\includegraphics[width=0.92\linewidth]{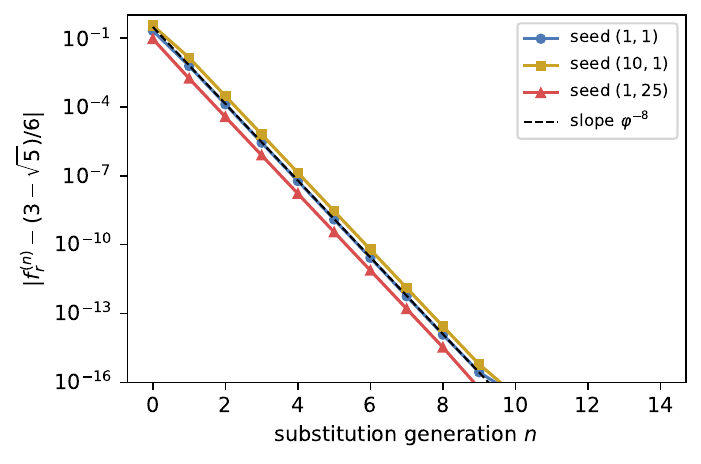}
\caption{Convergence of the reflected-hat fraction
$f_{r}^{(n)}$ to $(3-\sqrt{5})/6$ under iteration of the
census matrix $M$, independently of the starting configuration
(three unrelated seeds shown). The slope on the logarithmic scale is
the exact Perron--Frobenius rate $\gphi^{-8}$ of
Eq.~\eqref{eq:rate}.}
\label{fig:pf}
\end{figure}

\subsection{Strong local indistinguishability within each chirality
class}
\label{sec:hat-li}

The next theorem is the load-bearing statement of this section. It
holds for \emph{every} Hat tiling, with no restriction to a
substitution-generated subfamily; the unique-hierarchy theorem is
what makes this possible. Its correct scope, however, is \emph{per
chirality class}: patch frequencies are uniform within $\THat^{+}$
and within $\THat^{-}$, and the two classes are separated by them.

\begin{theorem}[Strong local indistinguishability, Hat, within each class]
\label{thm:LI}
Every finite legal patch $P$ possesses a well-defined frequency in
every Hat tiling, uniform over each chirality class:
\[
\freq_{T}(P) \;=\; \phi^{\pm}_{P}
\qquad \text{for all } T \in \THat^{\pm},
\]
with $\phi^{-}_{P} = \phi^{+}_{\bar P}$, where $\bar P$ denotes
the mirror image of the patch $P$. The two classes are distinct: for
chirality-sensitive patches $\phi^{+}_{P} \neq \phi^{+}_{\bar P}$,
the extreme case being single tiles, with ratio $\gphi^{4} : 1$.
Consequently, hypothesis (H1) of Section~\ref{sec:framework} holds
\emph{within each class}: the reduced density matrix $\rhoK^{\pm}$
on any bounded region $K$ depends on the class but not otherwise on
the representative $[T]$.
\end{theorem}

\begin{proof}
Fix a finite patch $P$ and $T \in \THat^{+}$. By the unique-hierarchy
theorem \cite[Thm.~5.1]{Smith2024hat}, for every $n$ the tiling $T$
decomposes uniquely into level-$n$ supertiles, drawn from the finite
set of level-$n$ supertile types generated by the metatile
substitution.

\emph{Step 1: counting inside supertiles.}
Let $c_{P}^{(n)}(\tau)$ denote the number of copies of $P$ contained
in the interior of a level-$n$ supertile of type $\tau$. Because the
substitution is deterministic, $c_{P}^{(n)}(\tau)$ depends only on
$(n, \tau)$, not on the ambient tiling. The census of supertile types
inside higher supertiles is governed by the primitive matrix of the
metatile substitution, so by the Perron--Frobenius theorem the ratio
$c_{P}^{(n)}(\tau) / \operatorname{area}(\tau^{(n)})$ converges as
$n \to \infty$ to a limit $\phi_{P}$ that is independent of $\tau$
(different types equidistribute at a rate governed by the subleading
eigenvalue, exactly as in Proposition~\ref{prop:PF}(5)).

\emph{Step 2: boundary copies are negligible.}
Copies of $P$ in $T$ that are \emph{not} interior to any level-$n$
supertile must lie within distance $\operatorname{diam} P$ of the
boundary network of the level-$n$ decomposition. Within a ball
$B_{R}$, the number of such positions is $O(R^{2} \gphi^{-2n})$,
vanishing relative to $R^{2}$ as $n \to \infty$.

\emph{Step 3: frequency.}
Combining Steps 1 and 2, for every $x$ and every $T$,
\[
\frac{N_{P}(T, B_{R}(x))}{\pi R^{2}}
\;\xrightarrow[R \to \infty]{}\; \phi_{P},
\]
with the convergence uniform in $x$: the count over $B_{R}$ is
sandwiched between supertile-interior counts at level $n(R)$ chosen so
that $\gphi^{2n} = o(R)$, plus the vanishing boundary correction. The
limit was defined without reference to $T$, so
$\freq_{T_{1}}(P) = \freq_{T_{2}}(P) =: \phi^{+}_{P}$ for all
$T_{1}, T_{2} \in \THat^{+}$.

\emph{Step 4: the mirror class.}
A tiling $T \in \THat^{-}$ is the mirror image of a tiling in
$\THat^{+}$ and composes uniquely into the mirrored metatile system;
applying Steps 1--3 to the mirrored substitution gives
$\freq_{T}(P) = \phi^{+}_{\bar P} =: \phi^{-}_{P}$. The two
frequency systems are genuinely different:
$\phi^{+}_{\mathrm{hat}} / \phi^{+}_{\mathrm{anti\text{-}hat}}
= \gphi^{4}$ by Proposition~\ref{prop:PF}(4), while the mirrored
ratio is $\gphi^{-4}$.

Finally, identical patch frequencies within a class are precisely the
strong local indistinguishability hypothesis of Li and Boyle
\cite{LiBoyle2024}, under which the reduced density matrices
coincide class-wise: $\rho_{K, [T]} = \rhoK^{\pm}$ for
$T \in \THat^{\pm}$.
\end{proof}

Theorem~\ref{thm:LI} is what makes a code possible at all: it is
hypothesis (H1), the statement that no bounded measurement can
identify the codeword. That it holds for \emph{every} Hat tiling ---
not merely for a substitution-generated subfamily --- is what lets
the code live on the full family; and its per-class form is not a
technical blemish but the origin of the sector structure of
\S\ref{sec:hat-code}.

\begin{corollary}
\label{cor:LI}
The reflected-tile frequency is $(3-\sqrt{5})/6$ in every tiling of
$\THat^{+}$ and $(3+\sqrt{5})/6$ in every tiling of $\THat^{-}$.
The reflection-symmetrized frequencies
$\freq_{T}(P) + \freq_{T}(\bar P)$ are uniform over all of
$\THat$.
\end{corollary}

\begin{lemma}[The classes are locally separated]
\label{lem:separation}
If $T^{+} \in \THat^{+}$ and $T^{-} \in \THat^{-}$, then $T^{+}$
and $T^{-}$ do not agree outside any bounded region.
\end{lemma}

\begin{proof}
Patch frequencies are density limits and are therefore unchanged by
the removal of a bounded set: the restriction $T|_{K^{c}}$ already
determines $\freq_{T}(P)$ for every $P$. If $T^{+}$ and $T^{-}$
agreed on $K^{c}$, the two tilings would share all patch
frequencies, contradicting the distinctness of the two frequency
systems in Theorem~\ref{thm:LI}.
\end{proof}

\newtext{The proof of Theorem~\ref{thm:LI} does not use the
identification of $\THat^{+}$ with the orbit closure of a substitution
tiling: patch frequencies are computed directly from the hierarchy.
The identification is nevertheless needed in Section~\ref{sec:hat-rec},
where the model-set description of \cite{BGS2025} is imported, and it
answers a question left open in the discovery paper, whether Hat
tilings with infinite ``fault lines'' exist outside the substitution
hull \cite[\S7]{Smith2024hat}.}

\begin{theorem}[Hull identification]
\label{thm:hull}
\newtext{Let $\OHat$ be the orbit closure, under
$\RR^{2}\rtimes C_{6}$, of any Hat tiling generated by the substitution
of \cite{Smith2024hat} with unreflected clusters. Then
$\THat^{+}=\OHat$: every Hat tiling whose clusters are unreflected
belongs to the hull, fault lines included. The mirrored statement holds
for $\THat^{-}$.}
\end{theorem}

\begin{proof}
\newtext{Let $T\in\THat^{+}$. By the unique-hierarchy theorem $T$
carries, for every $k$, a decomposition into level-$k$ supertiles, each
a copy of one of the four level-$k$ supertiles of the substitution,
realized by the same hats. Baake, G\"ahler and Sadun prove that the
metatile substitution forces the border \cite[Lemma~7]{BGS2025}: every
supertile, in every tiling by the metatiles obeying the matching rules,
determines a patch of hats covering the supertile together with a belt
of positive thickness around it, and under inflation the belt grows at
the same rate as the supertile. We have verified the first statement
independently at the level of hats, for the supertiles of level~3:
fixing the hats of one level-3 supertile and enumerating by exhaustive
SAT search all tilings by lattice-aligned hats and anti-hats of a disc
of radius $36$--$44$ kisrhombille units around it, every tiling
contains every hat of every metatile meeting the closed outline of the
supertile, for each of the four types (the certificates of
unsatisfiability for the negation, confirmed by two solvers, are
available from the authors on request). Since every Hat tiling is aligned to a single
kisrhombille lattice \cite[Lemma~A.6]{Smith2024hat}, this covers every
Hat tiling containing the supertile, not only those of the hull; the
growth of the belt with the level is taken from \cite{BGS2025}.
Now let $P$ be a finite patch of $T$, contained in a disc of radius
$r$. The regions of the level-$k$ supertiles of $T$ cover the plane,
so the centre of the disc lies in some level-$k$ supertile $s$; for
$k$ large enough that the forced belt is thicker than $2r$, the patch
$P$ lies inside the patch forced by $s$, which is a patch of the
substitution tiling and occurs in every tiling of $\OHat$. Hence every
finite patch of $T$ occurs in $\OHat$, that is, $T\in\OHat$.}
\end{proof}

\newtext{The finite enumeration of Proposition~\ref{prop:computational}
below cannot replace this argument: every configuration read off a
substitution patch belongs to the substitution language by
construction, so such an enumeration never sees a legal configuration
outside it.}

\subsection{Local recoverability}
\label{sec:hat-rec}

Local recoverability is where the Hat differs most instructively from
the Penrose tiling. Li and Boyle prove Penrose recoverability in one
stroke using Ammann lines: each of the five families of parallel
lines is visible outside $K$ and extends uniquely through it, and the
completed Ammann grid determines the tiling
\cite[App.~B.1]{LiBoyle2024}. The Hat has no known analogue of the
Ammann lines, so we argue instead along the lines of their
one-dimensional (Fibonacci) treatment: composition (deflation)
reduces a bounded erasure to a residual problem of uniformly bounded
size, and the residual problem is then analysed by other means. The
Fibonacci case also calibrates what can possibly be true: Li and
Boyle prove recoverability of Fibonacci quasilattices \emph{except
for the singular ones}, and the exception is unavoidable --- the two
singular Fibonacci strings are reflection-symmetric strings that
differ precisely in their two central digits
\cite[App.~B.3, footnote~3]{LiBoyle2024}, i.e.\ they are two distinct
legal tilings that agree outside a \emph{bounded} region. Any
correct recoverability theorem must therefore either exclude
singular tilings or prove that the Hat, unlike Fibonacci, has no
bounded singular pairs. Figure~\ref{fig:fibsing} shows the mechanism behind
the Fibonacci exception.

\begin{figure}[t]
\centering
\includegraphics[width=\linewidth]{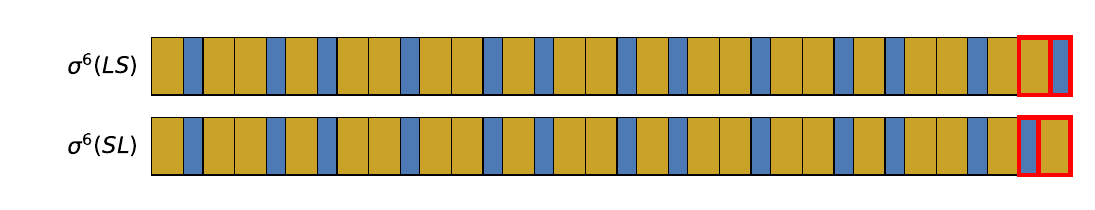}
\caption{Why recoverability can fail: the Fibonacci mechanism. Gold
cells are long tiles ($L$), blue cells short ($S$). The inflations of
the two-letter seeds $LS$ and $SL$ agree except in their final two
letters, at every order (shown: $\sigma^{6}$, verified
symbolically); in the limit this produces two distinct, perfectly
legal quasilattices that differ only in two central tiles
\cite[App.~B.3]{LiBoyle2024}. Erasing those two tiles is an
uncorrectable erasure. The open Conjecture~\ref{lem:defect} asks whether
the Hat admits such a bounded ``central swap'' or, like Penrose,
does not.}
\label{fig:fibsing}
\end{figure}

\subsubsection{The reduction}

\begin{proposition}[Reduction to bounded defects]
\label{prop:reduction}
\newtext{Suppose $T_{1},T_{2}\in\THat$ agree outside a disc of radius
$\rhoK$ about $x_{0}$, and let $T_{1}^{(n)},T_{2}^{(n)}$ be their
level-$n$ compositions. Let $r_{n}$ be the largest circumradius of a
level-$n$ supertile region ($r_{0}$ the circumradius of the hat). There
is a constant $N_{0}$, depending only on the Hat substitution, such
that for every $n\ge0$ the regions of the level-$n$ supertiles on
which $T_{1}^{(n)}$ and $T_{2}^{(n)}$ disagree lie in the disc of
radius $\rhoK+N_{0}\,r_{n}$ about $x_{0}$. In level-$n$ units the
radius of the differing region is therefore at most
$\rhoK/r_{n}+N_{0}$, which converges to $N_{0}$ independently of
$K$.} Moreover $T_{1}=T_{2}$ if and only if $T_{1}^{(n)}=T_{2}^{(n)}$.
\end{proposition}

\begin{proof}
\newtext{Call a level-$k$ supertile of $T_{1}$ \emph{common} if it is
also a level-$k$ supertile of $T_{2}$, with the same hats and the same
labels, and \emph{differing} otherwise; let $\rho_{k}$ be the radius of
the smallest disc about $x_{0}$ containing the regions of all differing
level-$k$ supertiles, so that $\rho_{0}\le\rhoK+2r_{0}$. By
\cite[Thm.~4.1]{Smith2024hat} the cluster containing a hat is
determined by the hats within two coronas of it, and by
\cite[Thm.~5.1]{Smith2024hat} the supertile containing a level-$k$
supertile is determined by the level-$k$ supertiles within a bounded
number of contact steps of it, the same number at every level $k\ge1$
because the supertiles are combinatorially equivalent to the
metatiles; let $c_{1}$ be the larger of the two radii. If every
level-$k$ supertile within contact distance $c_{1}$ of a level-$k$
supertile $c$ of $T_{1}$ is common, the composition rule, applied in
$T_{2}$ to the same tiles, assigns $c$ to the same level-$(k+1)$
supertile, which is therefore common. Hence every differing
level-$(k+1)$ supertile has a child within contact distance $c_{1}$ of
a differing level-$k$ supertile. A contact step between level-$k$
regions moves a point by at most $2r_{k}$, so that child lies within
$2(c_{1}+1)r_{k}$ of the disc of radius $\rho_{k}$, and the
level-$(k+1)$ region containing it within a further $2r_{k+1}$:
\[
\rho_{k+1}\;\le\;\rho_{k}+2(c_{1}+1)\,r_{k}+2r_{k+1}
\;\le\;\rho_{k}+2(c_{1}+2)\,r_{k+1},
\]
using $r_{k}\le r_{k+1}$ (the level-$(k+1)$ $H$ supertile contains a
level-$k$ supertile of every type). Summing,
$\rho_{n}\le\rhoK+2r_{0}+2(c_{1}+2)\sum_{k\le n}r_{k}$. Finally $r_{k}$
grows geometrically: the number of hats in a level-$k$ supertile is
governed by the Perron--Frobenius eigenvalue $\gphi^{4}$ of the
substitution matrix, so the area of a level-$k$ region, hence
$r_{k}^{2}$, is bounded below by a constant times $\gphi^{4k}$, while
by \cite[Thm.~1]{BGS2025} the Hat tiling is topologically conjugate,
up to a fixed linear map, to the self-similar CAP tiling, whose
level-$k$ supertiles have diameter proportional to $\gphi^{2k}$, and a
shape conjugacy displaces every point by a bounded amount
\cite{ClarkSadun2006}, so $r_{k}$ is also bounded above by a constant
times $\gphi^{2k}$. Thus $r_{k}\le C\gphi^{-2(n-k)}r_{n}$, the sum is
at most $C\gphi\,r_{n}$ (as $\sum_{j\ge0}\gphi^{-2j}=\gphi^{2}/\gphi=\gphi$),
and the claim follows with $N_{0}=2+2(c_{1}+2)C\gphi$.} The last
statement holds because composition is a bijection on tilings (unique
hierarchy).
\end{proof}

The reduction is the two-dimensional form of Li and Boyle's
deflation argument for Fibonacci \cite[App.~B.3 and
Prop.~3]{LiBoyle2024}, whose residual ambiguity (``a single
$01\leftrightarrow 10$ swap'') is the Fibonacci singular pair: a
bounded defect that survives deflation at every level. \newtext{The
residual radius $N_{0}$, measured in level-$n$ supertiles, is the room
in which such singular pairs live, when they exist. The first version
of this paper argued with a fixed Euclidean locality radius $R_{0}$ at
every level and obtained the fixed point $R^{*}=R_{0}/\gphi$; that
argument presupposed that the level-$n$ supertiles are scaled copies
of the metatiles, which they are not (Section~\ref{sec:hat-pf}), and
is withdrawn. The proposition above does not evaluate $N_{0}$; the
computations of Section~\ref{sec:hat-singular} are organized so that
an explicit value, once established, can be compared with them.}

\subsubsection{Recoverability of nonsingular Hat tilings}

The hull $\OHat$ is topologically conjugate to the hull of the CAP
tiling, which arises from a $4{:}2$ Euclidean cut-and-project scheme
with regular windows \cite{BGS2025}; parts of the window boundaries
are fractal \cite[Fig.~10]{BGMM2025}. As for every regular model
set, the system is almost automorphic: there is a continuous,
translation-equivariant \emph{torus parametrization}
$\beta : \OHat \to \mathbb{T}$ onto the compact abelian group
$\mathbb{T}$ underlying the cut-and-project scheme
\cite[Def.~4.2 and Prop.~4.3]{Schlottmann}, and $\beta$ is
one-to-one precisely over the (full Haar measure, flow-invariant) set
of \emph{nonsingular} parameters --- called \emph{generic} in
\cite{Schlottmann}, whose Eq.~(4.8) and Theorem~4.5 establish
exactly this almost-1-1 structure, together with unique ergodicity
and pure-point spectrum for every regular model set; see also
\cite[Ch.~7]{BaakeGrimm}. A tiling is called nonsingular if its
parameter is. Concretely, $\beta(T)$ is the tiling's
\emph{phase}: the position of the cut through the higher-dimensional
lattice that produces $T$. Any unbounded portion of a tiling pins its
phase exactly, while the content of a bounded window cannot shift it
--- and this rigidity is what the proof exploits.

\begin{theorem}[Recoverability, nonsingular case]
\label{thm:recover}
Let $T_{1}, T_{2} \in \OHat$ agree outside a bounded region $K$, and
suppose at least one of them is nonsingular. Then $T_{1} = T_{2}$.
\end{theorem}

\begin{proof}
\emph{Step 1: the parameters agree.} Let $d_{\mathbb{T}}$ be a
translation-invariant metric on $\mathbb{T}$. For $t \in \RR^{2}$
with $|t| \to \infty$, the translates $T_{1} - t$ and $T_{2} - t$
agree on the ball of radius $|t| - r_{K}$ around the origin, so
$d_{\Omega}(T_{1} - t,\, T_{2} - t) \to 0$ in the local tiling
metric. By uniform continuity of $\beta$,
$d_{\mathbb{T}}\bigl( \beta(T_{1} - t), \beta(T_{2} - t) \bigr) \to
0$. But $\beta$ intertwines the translation action with an isometric
action on $\mathbb{T}$, so this distance is \emph{independent of}
$t$:
\[
d_{\mathbb{T}}\bigl( \beta(T_{1}), \beta(T_{2}) \bigr)
 = d_{\mathbb{T}}\bigl( \beta(T_{1} - t), \beta(T_{2} - t) \bigr)
 \;\longrightarrow\; 0 ,
\]
whence $\beta(T_{1}) = \beta(T_{2})$.

\emph{Step 2: singleton fiber.} The common parameter
$\gamma = \beta(T_{1}) = \beta(T_{2})$ is nonsingular by hypothesis,
so the fiber $\beta^{-1}(\gamma)$ is a singleton and
$T_{1} = T_{2}$.
\end{proof}

Nonsingularity is invariant under the full isometry group (the fiber
structure is equivariant), so ``nonsingular Hat tilings'' is an
invariant family carrying full measure, on which both Li--Boyle
hypotheses now hold unconditionally. This is precisely the scope of
Li and Boyle's own Fibonacci code, which is likewise defined on
nonsingular quasilattices \cite[App.~B.3]{LiBoyle2024}. The theorem
is stated on $\OHat$, the hull of the right-handed class; the
mirrored statement covers $\THat^{-}$, and no cross-class case
arises: by Lemma~\ref{lem:separation}, tilings from different
chirality classes never agree outside a bounded region in the first
place.

Theorem~\ref{thm:recover} carries the construction from conditional
to unconditional: on nonsingular tilings both Li--Boyle hypotheses
now hold with no unproven input --- exactly the epistemic position of
the original Fibonacci code. Everything downstream, including the
sector structure of \S\ref{sec:hat-code}, stands on it; what
remains open is confined to the measure-zero singular set, to which
we now turn.

\subsubsection{The singular-pair question}
\label{sec:hat-singular}

What remains open is exactly one sharply posed question.

\begin{conjecture}[Bounded defect exclusion]
\label{lem:defect}
There is no pair of distinct Hat tilings whose symmetric difference
is nonempty and bounded.
\end{conjecture}

\newtext{We state this as a conjecture rather than a lemma: it is not
used in the proof of any result below, and the recoverability theorems
hold unconditionally on the nonsingular tilings without it.}
By Theorem~\ref{thm:recover} any such pair would consist of two
singular tilings over the same parameter; by
Proposition~\ref{prop:reduction} its difference region could be taken
of radius \newtext{at most $N_{0}$ in level-$n$ units} at every
composition level. Fibonacci
shows such pairs can exist in principle; Penrose shows they can be
absent. Because the hat needs no matching rules --- aperiodicity is
enforced by shape alone --- the question has a purely geometric
reformulation: \emph{does any region of the plane that is a union of
tiles in some Hat tiling admit two distinct tilings by hats?} Any
second tiling of such a region, substituted into the ambient tiling,
would immediately produce a bounded singular pair, and conversely.

\newtext{The reduction of Proposition~\ref{prop:reduction} turns the
question into a finite one as soon as $N_{0}$ is made explicit, and we
have prepared the finite side. Working on the level-7 substitution
patch ($372\,100$ hats, $142\,129$ metatiles, built from the
substitution of \cite{Smith2024hat} in exact lattice arithmetic), the
locality constants of the composition are $R_{0}=\sqrt{39}/2$ for
hats into metatiles and $\ell_{1}=\sqrt{3891}/8$ for metatiles into
level-2 supertiles (in kisrhombille units; a pattern of that radius
about a tile determines its parent and its position in it, and no
smaller radius does); every level-4 supertile with its touching
neighbours contains the disc of radius $\sqrt{5187}/4=18.005$ about
each of its hat centres; and every one of the $902$ congruence classes
of regions formed by the hats centred within $18$ units of a hat centre
of the patch admits exactly one tiling by hats and anti-hats
(exhaustive SAT enumeration, each unsatisfiability answer confirmed by
a second solver). Consequently a bounded singular pair whose
difference, after composition, fits in such a region does not exist;
what remains is an explicit bound of $18$ units on the residual
radius, which we have not established. If it holds, the Hat is
Penrose-like rather than Fibonacci-like in the sense of the preceding
paragraph: singular pairs exist, but they differ along unbounded sets,
such as the straight lines and wiggly paths of hats exhibited by
Tatham \cite[\S7.4]{Tatham2025}, and are harmless for erasure
correction.}

\subsubsection{Computational verification at small scale}

We verified computationally that no such region exists up to a
substantial scale, using two data sets of A.\ Ch\'eritat (vertex
coordinates of Hat tilings generated algorithmically). Both lie
exactly (to $2 \times 10^{-8}$) on the kisrhombille lattice
$\{(p/4,\, q\sqrt{3}/4) : p + q \text{ even}\}$.

The first set is a bare point cloud of $537$ tile corners with no
tile assignments --- a blind decoding problem. Taking the hat outline
from the discoverers' own software \cite{KaplanHatviz} and requiring
every corner of every placed tile to appear in the observed cloud, an
exact-cover computation decodes it into a \emph{unique} decomposition
of $66$ complete hats plus rim fragments, validating the entire
pipeline against data that carries no labelling.

The second set lists $34\,860$ vertices in $2\,490$ consecutive
$14$-vertex blocks; every block is verified congruent to the hat
$14$-gon, so this set \emph{is} a tiling patch of $2\,490$ hats
($19\,920$ kites, pairwise disjointness checked), containing the
first as a subset. Its reflected-tile fraction is
$313/2490 = 0.1257$, matching the asymptotic
$(3-\sqrt{5})/6 = 0.1273$ of Eq.~\eqref{eq:fr} within the
$\sim\!0.7\%$ statistical error of a sample of this size.

\begin{proposition}[Computational]
\label{prop:computational}
The region covered by the $2\,490$-hat patch admits exactly one
tiling by hats (exhaustive SAT enumeration over all $38\,342$
placements fitting inside the region). Consequently every one of the
$2^{2490}$ sub-regions formed as a union of its tiles admits exactly
one tiling by hats.
\end{proposition}

The consequence follows because an alternative tiling of any
sub-union could be substituted into the unique tiling of the full
region, producing a second one. The certified sub-regions include all
$39$ congruence classes of first-corona regions (a tile with its edge
neighbours) and all $35$ classes of radius-$2$ ball regions occurring
among the $2\,047$ interior tiles of the patch, as well as every
larger tile-union fitting within its extent (roughly $56$ tile
diameters). Hence any bounded singular pair for the Hat must have a
difference region not congruent to \emph{any} tile-union occurring in
this patch. Independent per-region SAT checks on the smaller data set
(all first coronas and radius-$2$ balls) agree. The computations use
only exact integer lattice arithmetic and exhaustive solution
enumeration; the certificates and the code that re-verifies them are
provided as ancillary files with this preprint, under an MIT licence.
Figure~\ref{fig:flower} illustrates one instance.

\newtext{The per-class records of the certificate also show that each
of the $66$ radius-$2$ configurations and each of the $30$ flower
configurations (a tile with its six neighbours) occurring in the patch
admits a unique completion within the patch. The first version of this
paper presented this as a finite adjacency check for the
identification $\THat^{+}=\OHat$; it is not one, for the reason given
after Theorem~\ref{thm:hull}, and the identification rests on that
theorem.}

\begin{figure}[t]
\centering
\includegraphics[width=\linewidth]{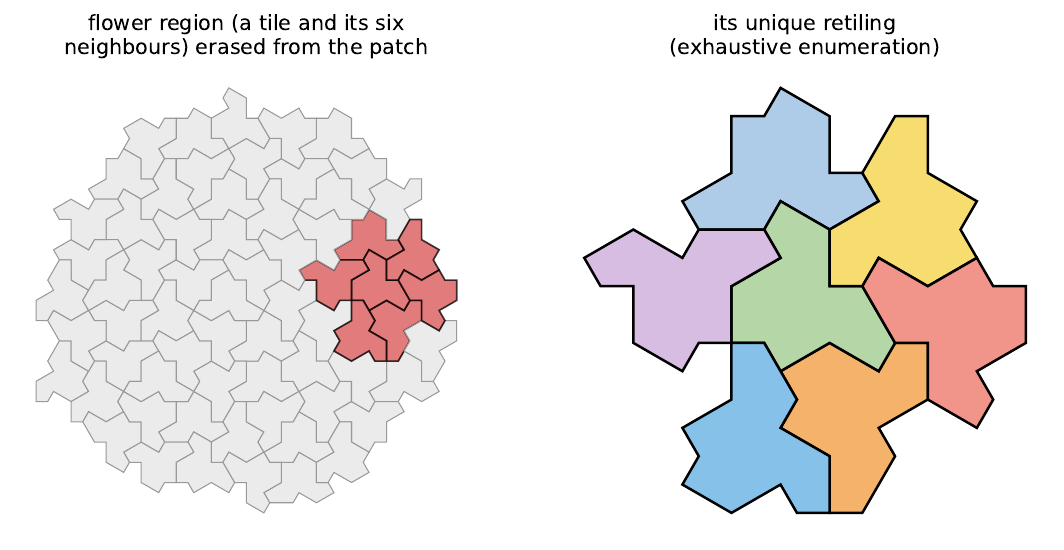}
\caption{One of the region-retiling tests, shown on the smaller
($66$-hat) data set. Left: a flower region (a
tile and its six neighbours) is erased from the decoded patch.
Right: exhaustive enumeration of all ways to tile the emptied region
by hats finds exactly one --- the original. Because the hat needs no
matching rules, any second tiling of any such region would
immediately produce two legal plane tilings agreeing outside a
bounded set; none exists in this sample.}
\label{fig:flower}
\end{figure}

\subsection{The Hat code: two sectors and a superselected chirality
bit}
\label{sec:hat-code}

Because proper isometries preserve tile chirality, the class
decomposition of \S\ref{sec:hat-pf} is respected by the gauge group:
$g\, \THat^{\pm} = \THat^{\pm}$ for every $g \in \SEtwo$, while
reflections exchange the two classes. Let
$\THat^{\pm,\mathrm{ns}}$ denote the nonsingular tilings of each
class ($\SEtwo$-invariant, of full measure). Define, for
$T \in \THat^{\pm,\mathrm{ns}}$,
\begin{align*}
|\Psi_{[T]}\rangle &\;=\; \int_{\SEtwo}\! dg\, |gT\rangle ,\\
\CHat^{\pm} &\;=\; \overline{\operatorname{span}}
\bigl\{ |\Psi_{[T]}\rangle : T \in \THat^{\pm,\mathrm{ns}}
\bigr\},
\end{align*}
and
\begin{equation}
\CHat \;=\; \CHat^{+} \oplus \CHat^{-} .
\label{eq:hatsum}
\end{equation}
The average runs over proper isometries only. This is the physically
natural choice --- position and orientation are unobservable in the
absence of an external frame, but parity is not a gauge symmetry ---
and it is the choice under which the two sectors are well defined.

\begin{theorem}[The Hat hybrid memory]
\label{thm:hatcode}
(i) Each sector $\CHat^{\pm}$ corrects erasure of any bounded
region $K \subset \RR^{2}$ (over the nonsingular tilings
unconditionally; over all of $\THat^{\pm}$ if
Conjecture~\ref{lem:defect} holds).
(ii) For every operator $O_{K}$ supported in $K$ and codewords
$|\Psi_{i}\rangle, |\Psi_{j}\rangle \in \CHat$,
\[
\langle \Psi_{i} | O_{K} | \Psi_{j} \rangle
\;=\; c_{s(i)}(O_{K})\, \delta_{ij} ,
\]
where $s(i) \in \{+,-\}$ is the sector label, and
$c_{+} \neq c_{-}$ for chirality-sensitive observables. The sector
label is therefore a \emph{superselected classical bit}: erasure of
any bounded region preserves it, \newtext{a window of a few tiles
reads it (see the discussion following Eq.~\eqref{eq:DeltaHat}),} and
coherences between the two sectors are not protected.
\end{theorem}

\begin{proof}[Proof sketch]
(i) Within a class, (H1) holds by Theorem~\ref{thm:LI} and (H2) by
Theorem~\ref{thm:recover}, so the Knill--Laflamme conditions
\eqref{eq:KL} hold sector-internally. (ii) For $i, j$ in opposite
sectors, every pair $gT_{i}, hT_{j}$ with $g, h \in \SEtwo$
consists of tilings from opposite classes (proper isometries preserve
the class), which by Lemma~\ref{lem:separation} never agree outside
$K$; the off-diagonal matrix element therefore vanishes exactly as in
the recoverability argument of Section~\ref{sec:framework}. The
diagonal constant depends only on the class by
Theorem~\ref{thm:LI}, and the observable below realizes
$c_{+} \neq c_{-}$.
\end{proof}

The distinguishing observable is the reflected-tile counter: for
bounded $K$,
\[
N^{\mathrm{refl}}_{K} \;=\;
\sum_{\text{tiles } t \subset K}
\mathbf{1}\bigl[\, t \text{ is an anti-hat} \,\bigr] .
\]
Tile chirality is invariant under every proper isometry, so
$N^{\mathrm{refl}}_{K}$ is well defined on all orbit-averaged
codewords, and by Corollary~\ref{cor:LI}
\begin{equation}
\begin{aligned}
\Delta_{\mathrm{Hat}}
&\;=\; \bigl\langle N^{\mathrm{refl}}_{K} \bigr\rangle_{-}
- \bigl\langle N^{\mathrm{refl}}_{K} \bigr\rangle_{+} \\
&\;=\; |K|_{\mathrm{tiles}}
\left( \tfrac{3+\sqrt{5}}{6} - \tfrac{3-\sqrt{5}}{6} \right)
\;=\; |K|_{\mathrm{tiles}}\, \frac{\sqrt{5}}{3} ,
\end{aligned}
\label{eq:DeltaHat}
\end{equation}
about $0.745\, |K|_{\mathrm{tiles}}$\newtext{: the two sectors have
different local statistics, and a window large enough for the counter
to concentrate classifies the sector with an error that vanishes as
$|K|$ grows. The bit can in fact be read exactly from a small window,
by the following local property of the two classes.}
Figure~\ref{fig:chirality} shows the handedness census of the
certified patch.

\newtext{In the level-7 substitution patch no two anti-hats share a
point: the closed polygons of any two anti-hats are at distance at
least $9/4$ kisrhombille units, while every $H$ metatile contains two
unreflected hats sharing an edge. Let $\Pi_{K}$ be the projector onto
the event that $K$ contains two tiles of equal handedness sharing an
edge. In a tiling of $\THat^{+}$ in which anti-hats are isolated such a
pair consists of two unreflected hats, in its mirror image of two
reflected ones, so the handedness of any one such pair identifies the
sector with certainty; every level-2 supertile contains an $H$
metatile, and in the patch every disc of radius $4.93$ units about a
hat centroid contains such a pair. The isolation of anti-hats is a
local property that we have verified on the patch and expect to hold
in every Hat tiling; we have not proved it, and the counter
$N^{\mathrm{refl}}_{K}$ remains the observable whose separation is
established.}

\begin{figure}[t]
\centering
\includegraphics[width=\linewidth]{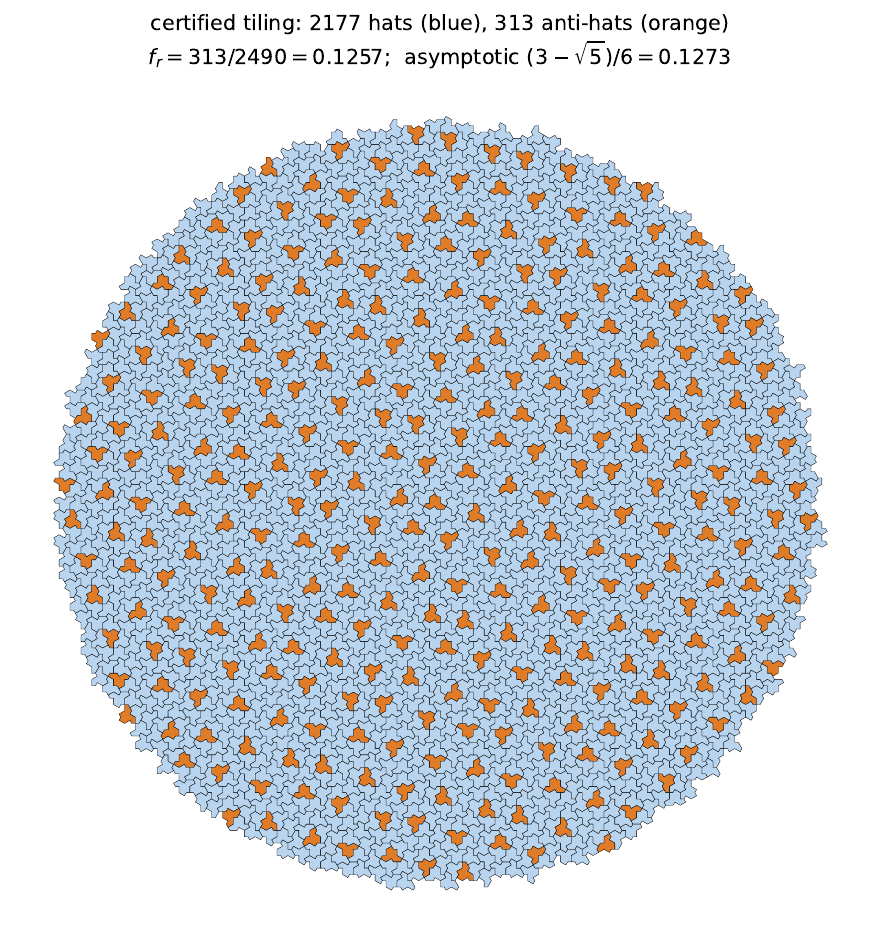}
\caption{Handedness census of the certified
$2490$-hat patch of \S\ref{sec:hat-rec}:
$2177$ hats (blue), $313$ anti-hats (orange); reflected fraction
$0.1257$ against the asymptotic $(3-\sqrt{5})/6 = 0.1273$. This
tiling belongs to $\THat^{+}$; in its mirror image the colours ---
and the value of the superselected bit --- are exchanged. Any bounded
window estimates the reflected fraction and thereby reads the bit,
with the macroscopic separation of Eq.~\eqref{eq:DeltaHat}.}
\label{fig:chirality}
\end{figure}

This structure is the precise opposite of topological degeneracy. In
a toric code the ground sectors are locally \emph{indistinguishable},
which is why relative phases between them are protected; here the
sectors are locally \emph{distinguishable}, which is why the label is
robust and classical. $\CHat$ is a hybrid quantum--classical memory:
one superselected bit --- the global handedness of the tiling's
long-range order --- alongside two erasure-correcting quantum
sectors. In channel terms: erasure of $K$ acts as a correctable
quantum channel on each sector and as a noiseless classical channel
on the label; the code is a direct sum of quantum codes indexed by a
decodable classical register. The same structure can be phrased in
the entanglement language of \cite{LiBoyle2024}: within a sector,
every codeword has the same reduced state, hence the same
entanglement entropy $S(\rhoK)$ for every region. Across sectors,
reflection is implemented by a unitary on the tiling Hilbert space (a
permutation of the basis of tilings), so for any reflection-symmetric
region $K$ --- a disk, say ---
$S(\rhoK^{+}) = S(\rhoK^{-})$: the chirality bit is completely
invisible to entanglement entropy, while being macroscopically
visible to the counter $N^{\mathrm{refl}}_{K}$. If reflections are gauged as well ($G = \Etwo$), the two
classes fall into a single orbit and the bit disappears;
\S\ref{sec:spec-qubit} shows that for the Spectre the same collapse
is caused already by a proper rotation.

The correctable-radius bookkeeping follows the deflation ladder: an
erasure of radius $r$ (in fundamental hat units $a$) closes after
$\lceil \log_{\gphi^{2}} (r/a) \rceil\newtext{{}+O(1)}$ deflation
steps\newtext{, the $O(1)$ being the residual radius $N_{0}$ of
Proposition~\ref{prop:reduction}}; equivalently, the correctable
radius after $n$ steps scales as
\begin{equation}
r_{n} \;\newtext{\asymp}\; a\, \gphi^{2n} .
\label{eq:rhat}
\end{equation}
The Hat tilings have $C_{6}$ statistical rotational symmetry,
enhanced to $D_{6}$ when reflections are counted, and first \v{C}ech
cohomology $\check{H}^{1}(\OHat, \ZZ) = \ZZ^{4}$ \cite{BGS2025};
these data enter the discussion of Section~\ref{sec:discussion}.

\section{The Spectre code}
\label{sec:spectre}

The Spectre is a strictly chiral aperiodic monotile discovered by
Smith, Myers, Kaplan and Goodman-Strauss shortly after the Hat
\cite{Smith2024spectre}. Unlike the Hat, every tiling by Spectres
uses tiles of a single handedness: reflections are unnecessary, and
the prototile set consists of a single shape in twelve rotational
orientations up to translation. Baake, G\"ahler, Maz\'a\v{c} and
Sadun \cite{BGMS2025} established that the twelve orientations split
into two subclasses of six related by a $30^{\circ}$ rotation, and
that the resulting long-range order comprises \emph{two}
local-indistinguishability classes rather than one. This parallels
the Hat's two chirality classes of \S\ref{sec:hat-pf}; the decisive
difference, developed in \S\ref{sec:spec-qubit}, is that the
isometry exchanging the Spectre's classes is a \emph{proper}
rotation, so the superselected label that survives $\SEtwo$-gauging
for the Hat does not survive it for the Spectre.

\subsection{Structural preliminaries}
\label{sec:spec-struct}

The Spectre is the equilateral member $\mathrm{Tile}(1,1)$ of the
Hat--Turtle continuum, with its 14 edges deformed so that reflections
cannot fit together in a tiling \cite{Smith2024spectre}. All tiles
occur at orientations that are integer multiples of $30^{\circ}$;
following \cite{Smith2024spectre}, we distinguish \emph{even}
Spectres (rotations by even multiples of $30^{\circ}$) from
\emph{odd} Spectres (odd multiples). Because Spectre tilings contain
no reflected tiles, the relevant symmetry group is the
orientation-preserving group $\SEtwo$ rather than $\Etwo$.

The structural input parallels the Hat exactly:
Theorem~2.2 of \cite{Smith2024spectre} states that the Spectre
``admits a tiling, and in any tiling it admits, each tile is
contained within an infinite, unique hierarchy of larger and larger
supertiles.'' This is the Spectre unique-hierarchy theorem, and it
plays below the role that \cite[\S 5]{Smith2024hat} plays in
Section~\ref{sec:hat}.

\paragraph{Substitution structure.}
Every Spectre tiling can be composed into non-overlapping copies of
two clusters \cite[Fig.~2.1]{Smith2024spectre}: a \emph{Spectre
cluster} of one Mystic and seven Spectres, and a \emph{Mystic
cluster} of one Mystic and six Spectres, where a Mystic is a
two-Spectre compound (one even, one odd). With $M_{n}$ Mystics and
$S_{n}$ single Spectres at level $n$, the rules
$S \mapsto M + 7S$, $M \mapsto M + 6S$ give
\begin{equation}
\begin{pmatrix} M_{n+1} \\ S_{n+1} \end{pmatrix}
= \begin{pmatrix} 1 & 1 \\ 6 & 7 \end{pmatrix}
\begin{pmatrix} M_{n} \\ S_{n} \end{pmatrix} .
\label{eq:Mspec}
\end{equation}
The matrix has trace 8 and determinant 1, so its eigenvalues are
$\lambda_{\pm} = 4 \pm \sqrt{15}$, with Perron--Frobenius eigenvalue
\begin{equation}
\lambda_{+} \;=\; 4 + \sqrt{15} \;\approx\; 7.873
\label{eq:lamspec}
\end{equation}
(area inflation per substitution step), matching the eigenvalue
computation of \cite{Smith2024spectre} and reappearing as the
squared-substitution cohomology eigenvalue in
\cite[Thm.~1]{BGMS2025}. The Perron--Frobenius eigenvector,
normalized to $v_{M} = 1$, is
\[
(v_{M},\, v_{S}) \;=\; \bigl( 1,\ 3 + \sqrt{15} \bigr),
\]
verified directly: $1 + (3+\sqrt{15}) = 4+\sqrt{15}$ and
$6 + 7(3+\sqrt{15}) = (4+\sqrt{15})(3+\sqrt{15})$.

\paragraph{Handedness and $\sigma^{2}$.}
The substitution reverses tile handedness at every step
\cite{Smith2024spectre}: applying $\sigma$ once to an ``even'' tiling
produces a cluster of ``odd'' tiles. The chirality-preserving
substitution is therefore $\sigma^{2}$, with matrix
$\bigl(\begin{smallmatrix} 7 & 8 \\ 48 & 55 \end{smallmatrix}\bigr)$,
Perron--Frobenius eigenvalue $(4+\sqrt{15})^{2} = 31 + 8\sqrt{15}$,
and \emph{linear} inflation $4+\sqrt{15}$ per $\sigma^{2}$ step. The
single-step linear inflation is $\sqrt{4+\sqrt{15}} \approx 2.806$.
Since the code superposes over orientation-preserving isometries
only, $\sigma^{2}$ is the natural level for correctable-radius
bookkeeping.

\paragraph{Even/odd Spectre frequencies.}
Each Mystic contains one even and one odd Spectre; each single
Spectre is even in the reference orientation of the substitution
rule. Counting individual Spectres,
$\#(\text{odd}) = v_{M}$, $\#(\text{even}) = v_{M} + v_{S}$,
$\#(\text{total}) = 2 v_{M} + v_{S}$, so the asymptotic
orientation-class frequencies are
\begin{equation}
\begin{aligned}
f_{\mathrm{odd}} &\;=\; \frac{5-\sqrt{15}}{10} \;\approx\; 0.1127,\\
f_{\mathrm{even}} &\;=\; \frac{5+\sqrt{15}}{10} \;\approx\; 0.8873,
\end{aligned}
\label{eq:fspec}
\end{equation}
with $f_{\mathrm{even}}/f_{\mathrm{odd}} = 4+\sqrt{15}$, in agreement
with \cite{Smith2024spectre}. This is the Spectre analogue of the
reflected-Hat frequency of Eq.~\eqref{eq:fr}, with the crucial
distinction that here the minority class is a \emph{rotational}
subclass rather than a reflected one. Figure~\ref{fig:spectrepatch}
verifies both statements in a generated tiling.

\begin{figure}[t]
\centering
\includegraphics[width=\linewidth]{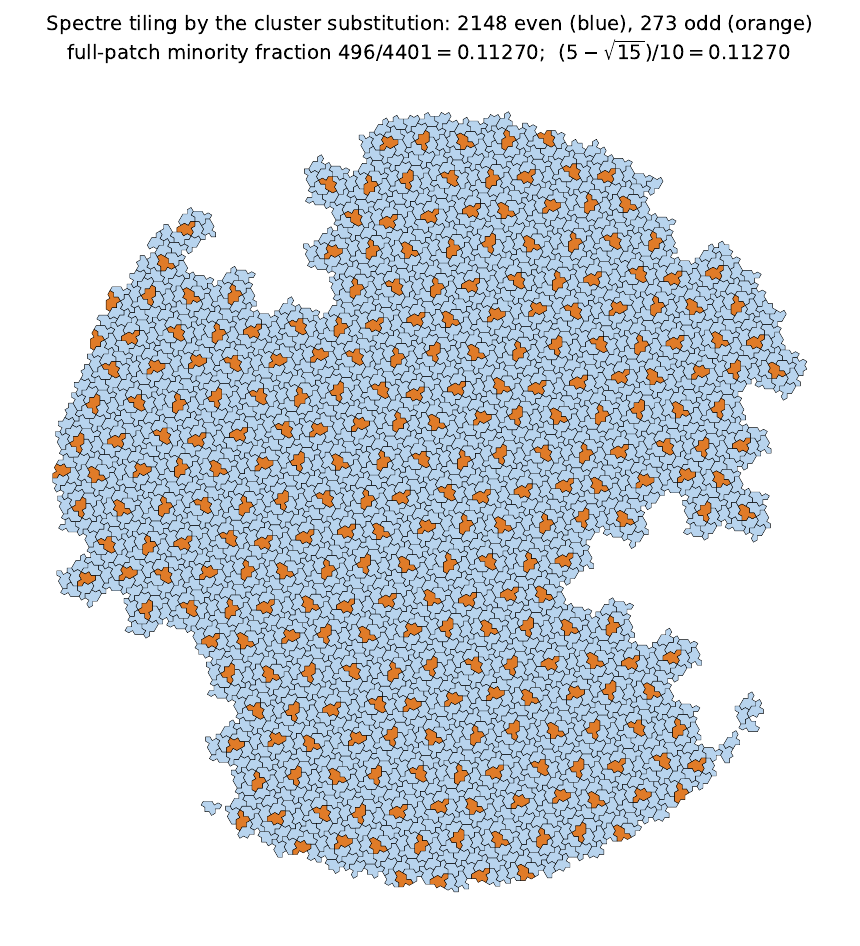}
\caption{A Spectre tiling generated by the $\Gamma$--$\Psi$ cluster
substitution of \cite{Smith2024spectre} (four inflation steps,
$4401$ tiles; central disk shown), coloured by orientation parity as
in Figure~\ref{fig:spectre}. Every placement in the patch is proper
(determinant $+1$, checked tile by tile), consistent with strict
chirality. The minority-orientation fraction of the full patch is
$496/4401 = 0.11270$, matching $(5-\sqrt{15})/10 = 0.11270$ of
Eq.~\eqref{eq:fspec} to five decimal places, and each of the six
even (odd) orientations appears with nearly equal frequency,
reflecting the $C_{6}$ statistical symmetry of each LI class.}
\label{fig:spectrepatch}
\end{figure}

\paragraph{Statistical symmetry and diffraction.}
Every finite patch occurs in six orientations related by rotations of
$2\pi/6$ with equal frequency, so Spectre tilings have $C_{6}$
statistical rotational symmetry but \emph{not} $C_{12}$: the even and
odd subclasses are not statistically interchangeable
\cite{BGMS2025}. The Spectre hull has pure-point dynamical spectrum
with a $4{:}2$ cut-and-project description and regular windows of
Rauzy-fractal type, and $\check{H}^{1}(\Omega, \CC) = \CC^{4}$
\cite[Thm.~1]{BGMS2025}. Pure-point diffraction is the quantitative
form of local indistinguishability required by the Li--Boyle
framework, and holds separately within each of the two LI classes.

\subsection{The two LI classes}
\label{sec:spec-li}

The central structural difference from the Hat is the following
result of \cite{BGMS2025}: the twelve Spectre orientations partition
into two subclasses of six (even and odd) related by a $30^{\circ}$
rotation, and there are correspondingly two distinct
local-indistinguishability classes, henceforth $\LI_{1}$ and
$\LI_{2}$. In $\LI_{1}$ the six even orientations form the
statistical majority, with total frequency $(5+\sqrt{15})/10$, and
the odd orientations the minority, with total frequency
$(5-\sqrt{15})/10$; $\LI_{2}$ is the image of $\LI_{1}$ under a
$30^{\circ}$ rotation, with the roles exchanged.
Figure~\ref{fig:spectre} shows the parity structure of the twelve
orientations.

\begin{figure}[t]
\centering
\includegraphics[width=0.9\linewidth]{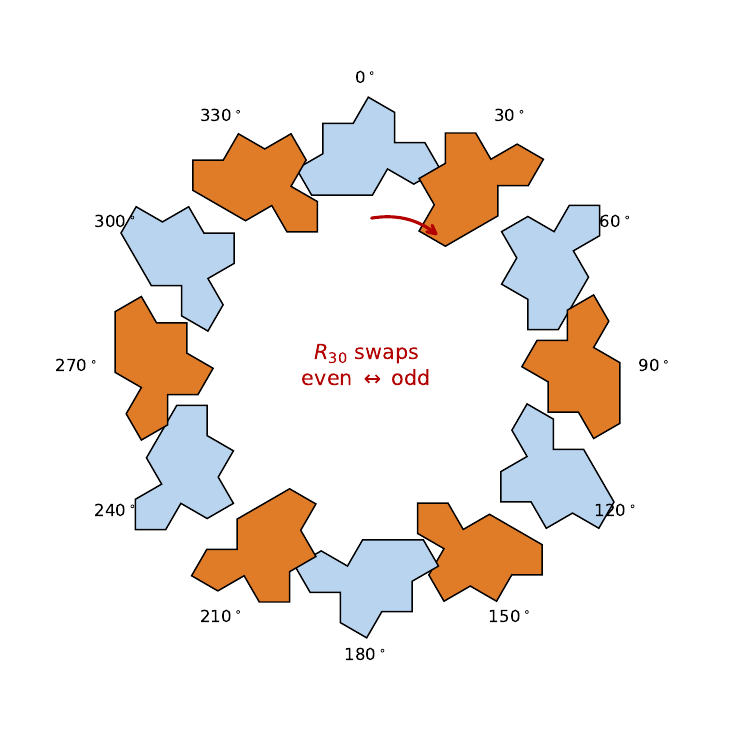}
\caption{The twelve Spectre orientations, coloured by orientation
parity: even multiples of $30^{\circ}$ blue, odd multiples orange
(the prototile itself is introduced in Figure~\ref{fig:anatomy}).
Within $\LI_{1}$ the blue orientations carry total frequency
$(5+\sqrt{15})/10$ and the orange ones $(5-\sqrt{15})/10$; the
rotation $R_{30}$ maps every blue to an orange, so the class label is
frame-dependent --- which is why it does not survive
$\SEtwo$-gauging (Theorem~\ref{thm:merger}).}
\label{fig:spectre}
\end{figure}

Rotating a tiling $T \in \LI_{1}$ by $30^{\circ}$ produces
$T' \in \LI_{2}$, but $T'$ is \emph{not} in the LI class of $T$:
every finite patch of $T'$ has a $30^{\circ}$-rotated counterpart in
$T$, but the frequency distribution over the twelve orientations
differs, so $T$ and $T'$ do not share the same patch statistics.
Both classes have pure-point dynamical spectrum
\cite{Queffelec, BGMS2025} and arise from the same cut-and-project
scheme with the same window, differing only in the projection
direction \cite{BGMS2025}.

Within each fixed class $\LI_{k}$, the argument of
Theorem~\ref{thm:LI} transfers verbatim, with the Spectre
unique-hierarchy theorem \cite[Thm.~2.2]{Smith2024spectre} in place
of \cite[Thm.~5.1]{Smith2024hat} and the census matrix of
Eq.~\eqref{eq:Mspec} in place of Eq.~\eqref{eq:M}:

\begin{theorem}[Strong local indistinguishability, Spectre, within each class]
\label{thm:LIspec}
For each $k \in \{1, 2\}$, every finite legal patch has a
well-defined frequency, identical in every tiling of $\LI_{k}$.
Consequently the reduced density matrix
$\rhoK^{(k)}$ on any bounded region is independent of the
representative tiling within $\LI_{k}$.
\end{theorem}

\subsection{Local recoverability within each class}
\label{sec:spec-rec}

The recoverability analysis of \S\ref{sec:hat-rec} transfers in both
of its parts. First, the reduction of
Proposition~\ref{prop:reduction}: composition of any Spectre tiling
into clusters exists, is unique \cite[Thm.~2.2]{Smith2024spectre},
and is local with a bounded contact radius\newtext{; the Spectre
supertiles are, like the Hat's, combinatorially but not geometrically
similar to the clusters \cite[App.~A]{Smith2024spectre}, and the
argument of Proposition~\ref{prop:reduction} applies with the
chirality-preserving step $\sigma^{2}$ and linear inflation
$\lambda=4+\sqrt{15}$ in place of $\gphi^{2}$: the differing region,
measured in level-$n$ units, stays bounded by a constant independent
of $K$}. Second, and more importantly, the nonsingular recoverability theorem:
the Spectre hull arises from a $4{:}2$ Euclidean cut-and-project
scheme with regular windows of Rauzy-fractal type \cite{BGMS2025},
so it too is almost automorphic, carries a torus parametrization,
and the proof of Theorem~\ref{thm:recover} applies verbatim within
each LI class. Hence \emph{nonsingular} Spectre tilings of a fixed
class are recoverable from the complement of any bounded region,
unconditionally; the singular-pair question
(Conjecture~\ref{lem:defect}) remains open for the Spectre exactly as for
the Hat, with the same purely geometric reformulation --- the
Spectre also tiles without matching rules. An erasure of radius $r$
closes after $\lceil \log_{4+\sqrt{15}} (r/a) \rceil\newtext{{}+O(1)}$
steps of $\sigma^{2}$-deflation:
\begin{equation}
r_{n} \;\newtext{\asymp}\; a\, (4+\sqrt{15})^{n} .
\label{eq:rspec}
\end{equation}
Recoverability holds tiling-by-tiling \emph{within} each LI class,
and across classes it is automatic: tilings from different classes
never agree outside a bounded region, by the frequency argument of
Lemma~\ref{lem:separation} applied to the class frequencies
\eqref{eq:fspec}. What becomes of the two classes at the level of
the \emph{code} depends on the gauge group, to which we now turn.

\subsection{Gauging: the Spectre sectors merge under rotations}
\label{sec:spec-qubit}

For the Hat, the two classes are exchanged by reflections, which lie
outside the gauge group $\SEtwo$; that is why the chirality bit of
\S\ref{sec:hat-code} exists. For the Spectre, the exchanging isometry
is the rotation $R_{30}$ by $30^{\circ}$ \cite{BGMS2025} --- a
\emph{proper} isometry. The consequence is immediate but decisive.

\begin{theorem}[Merger under $\SEtwo$]
\label{thm:merger}
Let $T \in \LI_{1}$. Then $R_{30} T \in \LI_{2}$ lies in the
$\SEtwo$-orbit of $T$. Hence the orbit superposition
$|\Psi_{[T]}\rangle = \int_{\SEtwo} dg\, |gT\rangle$ runs over
tilings of both classes, the putative sector spaces built from
$\LI_{1}$ and from $\LI_{2}$ coincide, and no superselected class
label survives $\SEtwo$-gauging.
\end{theorem}

\begin{proof}
$R_{30} \in \SEtwo$ and $\LI_{2} = R_{30}\, \LI_{1}$
\cite{BGMS2025}, so $[T]_{\SEtwo} = [R_{30}T]_{\SEtwo}$ meets both
classes, and the generating sets of the two putative sectors are
identical.
\end{proof}

The root cause is frame dependence. The Spectre label distinguishes
the two rotational subclasses --- even versus odd multiples of
$30^{\circ}$ --- and orientation parity is not an intrinsic property
of a tile: it refers to an external reference direction, which
rotations redefine. Tile chirality, by contrast, is frame-free. A
class label survives gauging by $G$ exactly when it is $G$-invariant;
chirality is $\SEtwo$-invariant, orientation parity is not.

The Spectre's two-sector structure is therefore relative to a smaller
gauge group. Let $G_{6} = \RR^{2} \rtimes C_{6}$ be the group
generated by all translations and the rotations by multiples of
$60^{\circ}$. Then $G_{6}$ preserves each LI class: each class has
$C_{6}$ statistical symmetry, and orientation parity is invariant
under $60^{\circ}$ rotations.

\begin{theorem}[Spectre sectors under the hexagonal subgroup]
\label{thm:spec-sector}
Define, for nonsingular $T \in \LI_{k}$,
$|\Psi^{(k)}_{[T]}\rangle = \int_{G_{6}} dg\, |gT\rangle$, let
$\CSpec^{(k)}$ be the closed span, and set
$\CSpec = \CSpec^{(1)} \oplus \CSpec^{(2)}$. Then:
(i) each sector corrects erasure of any bounded region
$K \subset \RR^{2}$ (over all of $\LI_{k}$ if the Spectre analogue
of Conjecture~\ref{lem:defect} holds);
(ii) tilings from different classes never agree outside a bounded
set, so cross-sector matrix elements of any $O_{K}$ vanish; and
(iii) the orientation-parity counter
\[
N^{\mathrm{even}}_{K} \;=\;
\sum_{\text{Spectres } t \subset K}
\mathbf{1}\bigl[ t \text{ is in an even orientation} \bigr] ,
\]
well defined because parity is $G_{6}$-invariant, separates the
sectors: by Eq.~\eqref{eq:fspec},
\begin{equation}
\Delta \;=\;
\bigl\langle N^{\mathrm{even}}_{K} \bigr\rangle_{(1)} -
\bigl\langle N^{\mathrm{even}}_{K} \bigr\rangle_{(2)}
\;=\; |K|_{\mathrm{tiles}}\, \frac{\sqrt{15}}{5} .
\label{eq:Delta}
\end{equation}
Relative to $G_{6}$, the Spectre code is thus a hybrid memory of
exactly the type of Theorem~\ref{thm:hatcode}.
\end{theorem}

The proof is identical in structure to Theorem~\ref{thm:hatcode}:
(H1) per class is Theorem~\ref{thm:LIspec}, (H2) per class is
\S\ref{sec:spec-rec}, class separation is the frequency argument of
Lemma~\ref{lem:separation} applied to the frequencies
\eqref{eq:fspec}, and the diagonal constants differ by
Eq.~\eqref{eq:Delta}.

To summarize the gauge dependence across the family: under $G_{6}$,
both monotiles carry a superselected bit; under $\SEtwo$, the Hat's
survives and the Spectre's is gauged away; under the full $\Etwo$,
both disappear (reflections exchange the Hat classes, and map Spectre
tilings to tilings by the mirrored Spectre, outside the family
altogether). Penrose, Ammann--Beenker and Fibonacci carry no bit
under any of these groups, having a single class each. Granting that
parity is not a gauge symmetry of nature, $\SEtwo$ is the physically
distinguished choice, and the Hat is the unique member of the family
whose long-range order stores a frame-independent classical bit.
Table~\ref{tab:gauge} summarizes the pattern.

\begin{table}[t]
\centering
\begin{tabular}{@{}lcc@{}}
\toprule
gauge group $G$ & Hat & Spectre \\
\midrule
$\RR^{2} \rtimes C_{6}$ & 2 sectors (bit) & 2 sectors (bit) \\
$\SEtwo$                 & \textbf{2 sectors (bit)} & 1 sector \\
$\Etwo$                  & 1 sector & 1 sector \\
\bottomrule
\end{tabular}
\caption{Sector structure as a function of the gauged isometries.
Enlarging $G$ can only merge sectors, and a merger occurs exactly
when $G$ acquires the class-swapping isometry: the proper rotation
$R_{30}$ for the Spectre (between the first and second rows), the
improper reflections for the Hat (between the second and third). The
physically natural row is $\SEtwo$, where the Hat keeps its bit and
the Spectre does not.}
\label{tab:gauge}
\end{table}

\subsection{Code parameters and comparison}
\label{sec:spec-params}

The parameters of the Spectre code: two translational LI classes
exchanged by a proper rotation, hence sectors only relative to
$G_{6}$; statistical rotational symmetry $C_{6}$ per class with no
reflection component; minority-orientation fraction
$(5-\sqrt{15})/10$ per class, a rotational subclass rather than a
set of reflected tiles; correctable radius
$r_{n} = a (4+\sqrt{15})^{n}$ per $\sigma^{2}$ step.
Table~\ref{tab:compare} collects the comparison across the Li--Boyle
family.

\begin{table*}[t]
\centering
\renewcommand{\arraystretch}{1.20}
\setlength{\tabcolsep}{6pt}
\small
\begin{tabular}{@{}lccccc@{}}
\toprule
                        & Penrose      & AB              & Fibonacci   & Hat                & Spectre \\
\midrule
Prototiles              & 2            & 2               & 2           & 1 (w/ refl.)       & 1 (chiral) \\
Linear inflation        & $\gphi$    & $1+\sqrt{2}$    & $\gphi$   & $\gphi^{2}$      & $4+\sqrt{15}$ ($\sigma^{2}$) \\
Rot.\ symmetry (per class) & $C_{5}$   & $C_{8}$         & $C_{1}$     & $C_{6}$            & $C_{6}$ \\
LI classes (transl.)    & 1            & 1               & 1           & \textbf{2}        & \textbf{2} \\
Class-swapping isometry & ---          & ---             & ---         & reflection         & $30^{\circ}$ rotation \\
Swap proper?            & ---          & ---             & ---         & no                 & yes \\
Sectors under $\SEtwo$ & 1            & 1               & 1           & \textbf{2}        & 1 \\
Superselected bit ($\SEtwo$) & no     & no              & no          & \textbf{yes}      & no \\
Minority fraction (per class) & ---   & ---             & ---         & $(3-\sqrt{5})/6$   & $(5-\sqrt{15})/10$ \\
Corr.\ radius $r_n$    & $\gphi^{n}$& $(1+\sqrt{2})^{n}$& $\gphi^{n}$& $\gphi^{2n}$   & $(4+\sqrt{15})^{n}$ \\
\bottomrule
\end{tabular}
\caption{The Li--Boyle code family. Both monotiles have two
translational LI classes; a class label survives $\SEtwo$-gauging
exactly when the class-swapping isometry is improper, which singles
out the Hat. The Spectre's sectors exist relative to the hexagonal
subgroup $G_{6}$ (Theorem~\ref{thm:spec-sector}). The Spectre is
combinatorially equivalent to a Hat--Turtle tiling on the
kisrhombille substrate \cite[Thm.~3.1]{Smith2024spectre}, on which
the code can be realized.}
\label{tab:compare}
\end{table*}

The last rows of Table~\ref{tab:compare} are the point of
Sections~\ref{sec:hat} and \ref{sec:spectre}: the monotiles are the
first members of the family with more than one LI class, and among
them only the Hat's class label is exchanged by an improper isometry
--- which is why, once all proper isometries are gauged, the Hat
alone retains a superselected bit.

\section{Code dimension and open problems}
\label{sec:discussion}

\paragraph{What does the code dimension count?}
A natural finite-size question is how the dimension of a
toroidal-approximant code space grows with system size $L$. We
caution against a tempting but unfounded shortcut: the substitution
gives $N_{\mathrm{tiles}} \sim \lambda_{+}^{\,n}$ with
$L \sim \lambda_{\mathrm{lin}}^{\,n}$, whence
$N_{\mathrm{tiles}} \sim L^{2}$ --- but this is an area identity that
counts tiles, not code sectors, logical states, or inequivalent
approximants, and no theorem converts one into the other. A more
promising mechanism is topological: the Hat hull has
$\check{H}^{1}(\OHat, \ZZ) = \ZZ^{4}$ \cite{BGS2025} and the Spectre
hull $\check{H}^{1}(\Omega, \CC) = \CC^{4}$ \cite{BGMS2025}, and in
periodic topological codes the ground-space degeneracy is governed by
exactly such first-cohomology data (for the toric code,
$H^{1}(T^{2}) = \ZZ^{2}$ yields four ground states). Whether the four
generators of $\check{H}^{1}$ manifest as logical operators, flux
sectors, or finite-size sector labels of the monotile codes is, to
our knowledge, unexplored; we leave the finite-size dimension as an
open problem rather than assert a scaling law.

\paragraph{Open problems.}
(i) Settle the singular-pair question (Conjecture~\ref{lem:defect}) for
the Hat and the Spectre: does either monotile admit two distinct
tilings agreeing outside a bounded region, as Fibonacci does, or
none, as Penrose does? Equivalently: can any tile-union region be
retiled a second way? Our computations exclude all tile-union regions occurring
in a certified $2490$-tile sample\newtext{, and those of
Section~\ref{sec:hat-singular} every region of radius $18$
kisrhombille units occurring in a level-7 substitution patch}; scaling
further via substitution-generated patches, or finding a structural
obstruction (a Hat analogue of the Ammann lines, or an argument
exploiting the border-forcing property of the CAP inflation
\cite[Fig.~9]{BGMM2025}), would settle it. \newtext{An explicit value
of the residual radius $N_{0}$ of Proposition~\ref{prop:reduction}, in
hats, would make the second check decisive.}
(ii) Relate the rank-4 first cohomology to the code structure, as
sketched above.
(iii) Construct explicit finite toroidal approximants for the Hat and
Spectre codes and determine their sector counts, connectivity under
local moves, and the fate of the Hat chirality bit at finite
size.
(iv) Physical realizations: Hat tilings live on the periodic
kisrhombille substrate, and the Spectre is combinatorially equivalent
to a Hat--Turtle tiling on the same substrate
\cite[Thm.~3.1]{Smith2024spectre}, suggesting lattice-model
implementations. For the Hat, the chirality bit of
\S\ref{sec:hat-code} would appear as a two-fold ground-state label
read by the coarse local observable of Eq.~\eqref{eq:DeltaHat}.

\section*{Acknowledgments}
J.B. appreciates fruitful discussions with Maria Vallespir-Socias,
Joana Rossell\'o, Margalida Batle, Maria del Mar Batle, and Regina
Batle. J.B. is grateful to Arnaud Ch\'eritat (Institut de
Math\'ematiques de Toulouse, Paul Sabatier University) for kindly
providing a big sample of hat monotiles. \newtext{We appreciate great insight from Juan Carlos Romero.} J.B. received no
funding for the present work.

\end{document}